\documentclass[conference]{IEEEtran}
\IEEEoverridecommandlockouts

\usepackage{cite}
\usepackage{amsmath,amssymb,amsfonts}
\usepackage{algorithmic}
\usepackage{graphicx}
\usepackage{textcomp}
\usepackage{xcolor}
\usepackage{amsthm}
\def\BibTeX{{\rm B\kern-.05em{\sc i\kern-.025em b}\kern-.08em
    T\kern-.1667em\lower.7ex\hbox{E}\kern-.125emX}}

\newtheoremstyle{mythmstyle} 
    {3pt} 
    {3pt} 
    {\itshape} 
    {} 
    {\bfseries}
    {.} 
    {0.5em} 
    {\textbf{#1~#2}}
\theoremstyle{mythmstyle}
\newtheorem{theorem}{Theorem}
\newtheorem{lemma}{Lemma}
\newtheorem{definition}{Definition}
\newtheorem{assumption}{Assumption}
\newtheorem{remark}{Remark}

\begin{document}

\title{An early termination strategy for the distributed biased min-consensus protocol under disturbances

\thanks{Zicheng Huang and Wangzhi Zhou contributed equally to this work.
This work was supported by
the National Science and Technology Major Project of China under Grant
No. 2022ZD0120003, the National Natural Science Foundation of China
under Grant No. 62303112,
and the Natural Science Foundation of Jiangsu Province under Grant No.
BK20230826. (Corresponding author: Yuanqiu Mo.)}
}

\author{
\IEEEauthorblockN{Zicheng Huang}
\IEEEauthorblockA{
\textit{School of Mathematics} \\
\textit{Southeast University} \\
Nanjing, China \\
213233132@seu.edu.cn}
\and

\IEEEauthorblockN{Wangzhi Zhou}
\IEEEauthorblockA{
\textit{School of Mathematics} \\
\textit{Southeast University} \\
Nanjing, China \\
213230425@seu.edu.cn}
\and

\IEEEauthorblockN{Yuanqiu Mo}
\IEEEauthorblockA{\textit{School of Mathematics} \\
\textit{Southeast University} \\
Nanjing, China\\
yuanqiumo@seu.edu.cn}
}

\maketitle

\begin{abstract}
The distributed biased min-consensus (DBMC) protocol is an iterative scheme that solves the shortest path problem asymptotically, requiring only local information exchange between neighboring nodes. By appropriately designing the gain function, prior work~\cite{article3} proposed a DBMC-based system that ensures convergence within a pre-specified time interval. However, this guarantee assumes the absence of disturbances. In this paper, we study the DBMC-based system under disturbances affecting the edge weights. We first establish rigorous error bounds on the resulting state estimates. Building on this analysis, we then propose a practical early termination strategy to prevent potential singularities, specifically, unbounded gain, that may arise in the presence of disturbances, while still ensuring that the shortest paths are correctly identified. Simulations are performed to validate and illustrate the theoretical results.
\end{abstract}


\section{Introduction}
The shortest path problem seeks to determine the path with the smallest total weight between two nodes in a graph~\cite{bellman1958}. While classical algorithms such as Dijkstra’s algorithm~\cite{Dijkstra}, Bellman-Ford algorithm~\cite{bellman1958} and A* algorithm~\cite{b7} solve this problem effectively in centralized settings, their reliance on global information limits scalability and resilience in large or failure-prone networks.  
As a distributed alternative, the distributed biased min-consensus (DBMC) algorithm proposed in~\cite{article2} offers a more flexible and scalable solution. By relying solely on local communication, DBMC allows each node to gradually identify its shortest path without requiring global topology knowledge or carefully chosen initial values. This decentralized nature makes the method particularly well-suited for applications in dynamic, large-scale, or resource-constrained environments.

Given these advantages, theoretical investigations of the DBMC protocol have produced extensive results on its convergence and robustness. For discrete-time implementations, prior studies~\cite{discrete1,article1,discrete2,discrete3,discrete5} have established finite-step convergence with explicit iteration bounds based on graph parameters. Robustness under persistent disturbances has also been established via Lyapunov analysis~\cite{discrete4}. In the continuous-time setting, global asymptotic and regional exponential stability have been established using LaSalle’s invariance principle~\cite{article2} and non-smooth Lyapunov functions~\cite{b15}, respectively. However, these results do not ensure convergence within a user-defined time frame, which is often required in hierarchical or real-time applications such as robot path planning for obstacle avoidance~\cite{discrete2}, or route computation in content delivery networks~\cite{b16}. 

To address this, pre-specified finite time (PT) control has been proposed, which ensures exact convergence within a fixed duration, regardless of initial conditions or design parameters~\cite{b17}. Compared to traditional finite time methods, PT control achieves faster and more predictable convergence by applying time-varying scaling functions~\cite{b19}, but may suffer from instability under disturbances~\cite{article4,b21}. To improve robustness, practical pre-specified finite time (PPT) control has been developed~\cite{b22}, which relaxes exact convergence by requiring states to remain within a neighborhood of the target before the deadline. Similar to PT control strategies, this is typically achieved by incorporating time-base generators (TBGs)~\cite{b23}.

Our work builds on the recent study~\cite{article3}, which introduced a PT control strategy to ensure the convergence of DBMC within a user-defined time frame. This
was achieved by employing a time-varying gain instead of
a constant gain used in the nominal DBMC, which becomes
unbounded at a pre-specified time instant.  Since the state error also vanishes at this instant, it has been established that the DBMC dynamics remain bounded and continuously differentiable under the PT control strategy \cite{article3}. However, in the presence of disturbances, the state of DBMC may deviate from the stationary point, causing the amplifying gain to destabilize the system near the prescribed deadline, leading to erratic DBMC dynamics. To address these issues, this paper proposes terminating the PT-controlled DBMC at a time strictly before the prescribed deadline, thereby avoiding unstable behaviors under disturbances while still guaranteeing the shortest path. Towards this end, we first derive refined error bounds for the perturbed DBMC protocol. We then establish a relationship between error bounds and correct path identification, and finally propose a practical early termination strategy that ensures accurate path selection.

In the rest of the paper, Section \ref{sec:pre} introduces the necessary graph-theoretic background and reviews the pre-specified finite time DBMC protocol proposed in~\cite{article3}; Section~III presents the main theoretical results; Section~IV provides supporting simulations; and Section~V concludes this article.

\section{Preliminaries}\label{sec:pre}

\subsection{Graph Theory}

Consider a directed graph \( G = (V, E) \), where \( V = \{0, 1, \dots, n\} \) is the set of nodes, and \( E \subseteq V \times V \) is the set of directed edges. Each edge \( (i, j) \in E \) indicates that node \( i \) can move directly to node \( j \), and is associated with a positive weight \( w_{ij} > 0 \). For each node \( i \), we define its out-neighbor set as \( \mathcal{N}_i = \{ j \in V \mid (i,j) \in E \} \).

A path from node \( i \) to node \( j \) is a finite sequence of nodes \( \{i_0, i_1, \dots, i_\ell\} \) such that (i) \( i_0 = j \), \( i_\ell = i \); (ii) \( (i_k, i_{k-1}) \in E \) for all \( 1 \leq k \leq \ell \); and (iii) all nodes in the sequence are distinct. The length of a path is defined as \( \sum_{k=1}^\ell w_{i_k i_{k-1}} \). Self-loops are not allowed, i.e., no node is connected to itself by an edge.


Let \( S_1 \subsetneqq V \) be a nonempty set of source nodes, and let \( S_2 = V \setminus S_1 \) be a nonempty set of non-source nodes. Our objective is to find a path from each node \( i \) to some node in \( S_1 \) with minimal length.

\begin{assumption}
\label{ass:Graph}
Graph G is directed, and for all \( i \in V \), there exists a directed path from \( i \) to some \( j \in S_1 \).
\end{assumption}

Under Assumption~\ref{ass:Graph}, we define the length of the shortest path from node \( i \) to its nearest source as \( p_i \). Then \( p_i \) satisfies the following recursive relation according to Bellman's optimality principle~\cite{bellman1958}:
\begin{align}
p_i =
\begin{cases}
0, & i \in S_1, \\
\min\limits_{j \in \mathcal{N}_i} \{ p_j + w_{ij} \}, & i \in S_2.
\end{cases}
\label{eq:bellman}
\end{align}

To further support the analysis in later sections, we introduce the following structural definitions related to the graph.

\begin{definition}
A true parent node of node \( i \) is an out-neighbor \( j \in \mathcal{N}_i \) such that \(p_j + w_{ij} \) in (\ref{eq:bellman}) is minimized. Since \(i\) may have multiple parents, we denote the set of all such nodes by \( \mathcal{P}_i \).
\label{def:parent}
\end{definition}
\begin{definition}\label{def:diameter}
Given a graph \( G = (V, E) \), its effective diameter \( \mathcal{D}(G) \) is defined as the maximum number of nodes in any sequence ending at a source node, where each node in such a sequence is a true parent of its predecessor.
\end{definition}

\subsection{Problem Statement}
We consider the DBMC protocol under the PT control strategy proposed in~\cite{article3}, operating over the time interval \( [0, T_s) \). From~\cite{article3}, the system dynamics are defined as
\begin{align}
\dot{x}_i(t) =
\begin{cases}
0, & i \in S_1, \\
-\eta(t)\left(x_i(t) - \min\limits_{j \in \mathcal{N}_i} \{ x_j(t) + w_{ij} \} \right), & i \in S_2,
\end{cases}
\label{eq:nominal-DBMC}
\end{align}
where \( x_i(t) \) is the state of node \( i \), and \( \eta(t) \) is a time-varying gain function defined as
\begin{align}
\eta(t) = \gamma + 2\frac{\dot{\mu}(t)}{\mu(t)}, \quad
\mu(t) = \left(\frac{T_s}{T_s - t}\right)^{1+h},
\end{align}
where \( h>-1/2 \) and \( \gamma >0\) are the design parameters.

Similar to \cite{article3}, we need the following assumption.
\begin{assumption}
\label{ass:overestimated}
The initial states of the non-source nodes are overestimated, that is, \( x_i(0) \geq p_i \) for all \( i \in S_2 \). Moreover, \( x_i(0) = 0 \) for all \( i \in S_1 \).
\end{assumption}

As shown in~\cite{article3}, under Assumptions~\ref{ass:Graph} and~\ref{ass:overestimated}, \(x_i(t)\) in \eqref{eq:nominal-DBMC} converges to \(p_i\), the length of the shortest path from \(i\) to its nearest source in a finite time. Specifically, 
\begin{align}
\!\!\! \lim_{t \to T_s^-} x_i(t) = p_i, \ \lim_{t\to T_s^-}\dot x_i(t)=0.
\end{align}

However, \cite{article3} considers an ideal disturbance-free scenario. In this paper, we extend the above framework by introducing time-varying edge weights of the form \( w_{ij}(t) = w_{ij} + u_{ij}(t) \), where \( u_{ij}(t) \) denotes an asymmetric additive disturbance that further satisfies the following condition.

\begin{assumption}
\label{ass:bound}
The disturbance \( u_{ij}(t) \) is continuous and bounded for all \( (i,j) \in E \), satisfying
\begin{align}
-w_{ij} < -u_{ij}^{-} \leq u_{ij}(t) \leq u_{ij}^{+},
\end{align}
where \( u_{ij}^{-}, u_{ij}^{+} \geq 0 \) are the disturbance bounds. 
\end{assumption}
We use \(u^-, u^+ \geq 0\) to denote uniform bounds of \(u_{ij}(t)\) in Assumption~\ref{ass:bound}, such that
\begin{equation}
- u^- \leq -u^-_{ij}\leq u_{ij}^+ \leq u^+, \quad \forall i \in V,\, j \in \mathcal{N}_i.
\label{eq:uniform}
\end{equation}

In the presence of edge-weight perturbations, \eqref{eq:nominal-DBMC} is now interpreted as:
\begin{align}
\dot{x}_i(t) =
\begin{cases}
0, & i \in S_1, \\-\eta(t)\left(x_i(t) - \min\limits_{j \in \mathcal{N}_i} \{ x_j(t) + w_{ij}(t) \} \right), & i \in S_2.
\end{cases}
\label{eq:pert-DBMC}
\end{align}

We define \(e_i(t) = x_i(t) - p_i\) to measure the deviation from \( p_i \). Naturally, \( \dot{e}_i(t) = \dot{x}_i(t) \). The primary objective of this paper is not to eliminate the deviation \( e_i(t) \), but to identify the shortest path for each non-source node by terminating the disturbed DBMC~\eqref{eq:pert-DBMC} at a specific time before the predefined time instant. In doing so, we avoid infinite gain, even when \( e_i(t) \) has not yet vanished. Note that Assumption~\ref{ass:overestimated} can now be equivalently expressed as \( e_i(0) \geq0\) for all \(i \in S_2 \) and \(e_i(0)=0\) for all \(i\in S_1\).

Another definition is required to facilitate our subsequent analysis. 
\begin{definition}
A current parent node of node \( i \) is an out-neighbor \( j \in \mathcal{N}_i \) that minimizes \( x_j(t) + w_{ij} + u_{ij}(t) \) in (\ref{eq:pert-DBMC}). Since such \(j\)  may not be unique, we use \( \mathcal P_i(t) \) to denote the set of all current parent nodes.
\label{def:currentparent}
\end{definition}

\section{Main Results}

We now present our analysis of the perturbed DBMC~\eqref{eq:pert-DBMC} over the time interval \( 0 \leq t < T_s \). Assumptions~\ref{ass:Graph}--\ref{ass:bound} are assumed to hold throughout this section.

\begin{remark}
\label{re:path}
 Since Assumption 1 guarantees that every node is connected to the source set \( S_1 \), it follows that for each \(i\in S_2\), there exists a sequence of nodes \(\{i_0, i_1, \dots, i_\ell\}\) that forms the shortest path from node \(i\) to its nearest source such that: (i) \( i_0 \in S_1\), \( i_\ell = i\); (ii) \( \ell\leq {\mathcal D}(G)-1\); (iii) \(i_k \in \mathcal{P}_{i_{k+1}}\) for all \(0 \leq k \leq \ell-1\); and (iv) \(\sum_{k = 1}^{\ell } w_{i_{k}i_{k - 1}} = p_i\), with \(p_i\) defined in (\ref{eq:bellman}).
\end{remark}
This observation implies that the analysis of any node \(i \in S_2\) can be reduced to examining the evolution of the nodes along a node sequence ending at \(i\), as described in Remark~\ref{re:path}. Therefore, \emph{in the subsequent analysis, we conduct analysis on the sequence of nodes as in Remark~\ref{re:path} and use \(i_k\) to index the node therein.}

Unless explicitly stated otherwise, all proofs of the stated results are provided in the Appendix.

\subsection{Error Bounds under Bounded Disturbances}

We first introduce an auxiliary function \( \phi(t) \) defined by
\begin{align}
\phi(t) =
\exp\left(\int_0^t \eta(s)\, \mathrm{d}s\right) = e^{\gamma t} \left(\frac{T_s}{T_s - t}\right)^{2 + 2h},
\label{eq:phi}
\end{align}
which streamlines the notation in subsequent computations.

The following lemma provides an upper bound for \(e_{i_\ell}(t)\).
\begin{lemma}\label{lem:upper}
Consider the sequence of nodes introduced in Remark~\ref{re:path}. The state error of each node in this sequence obeys
\begin{align}
e_{i_\ell}(t) \leq \mathcal{E}_{i_\ell}(t) + \sum_{k = 0}^{\ell} u_{i_k i_{k-1}}^{+},~ \forall \ell \in \{0, \cdots, \mathcal{D}(G) -1 \}
\label{eq:lem1:eq1}
\end{align}
where \(u_{i_0i_{-1}}^{+} := 0\), \(u_{i_k i_{k-1}}^{+}\) with \(k \in \{1, \cdots, \mathcal{D}(G) - 1 \}\) is defined in Assumption~\ref{ass:bound}, \(\mathcal{D}(G)\) is defined in Definition~\ref{def:diameter}, and 
\begin{align}
\mathcal{E}_{i_\ell}(t) = \sum_{k = 0}^{\ell} \frac{e_{i_k}(0)(\ln \phi(t))^{\ell-k}}{\phi(t)(\ell-k)!}. 
\label{eq:mathcal-E}
\end{align}
\end{lemma}

\begin{remark}
\label{re:mathcal-E}
The term \( \mathcal{E}_{i_\ell}(t) \) denotes the upper bound of \( e_{i_\ell}(t) \) in the absence of edge-weight disturbances. Since \( \phi(t) \to +\infty \) as \( t \to T_s^{-} \), and the numerator in (\ref{eq:mathcal-E}) grows only logarithmically in \( \phi \) while the denominator grows polynomially, it follows that \( \mathcal{E}_{i_\ell}(t) \to 0 \) as \( t \to T_s^- \).
\end{remark}

We next establish a lower bound when all disturbances are non-negative.

\begin{lemma}\label{lem:lower}
Consider (\ref{eq:pert-DBMC}). If \( u_{ij}(t) \geq 0 \) for all \( (i, j)\in E \) with \(u_{ij}(t)\) defined in Assumption~\ref{ass:bound}, then \(e_i(t)\geq 0\) for all \(i\in V\).
\end{lemma}

\begin{proof}
The proof follows from applying the comparison principle~\cite{khalil2002} to Lemma~4 in~\cite{article3}, and is omitted here.
\end{proof}

We now turn to the scenario where negative-weight disturbances may occur. First, consider \( u_{ij}(t) \) bounded proportionally to the corresponding edge weight \( w_{ij} \).

\begin{theorem}
\label{thm:proportional}
Consider (\ref{eq:pert-DBMC}). Suppose that \( u_{ij}(t) \), as defined in Assumption~\ref{ass:bound}, satisfies \( -\alpha_1 w_{ij} \leq u_{ij}(t) \leq \alpha_2 w_{ij} \) with \( 0 \leq \alpha_1, \alpha_2 < 1 \). Then, for all \( \ell \in \{1, \cdots, \mathcal{D}(G) - 1 \} \),
\begin{align}\label{eq:bound}
-\alpha_1  p_{i_\ell} \leq e_{i_\ell}(t) \leq  \mathcal E_{i_\ell}(t)+\alpha_2  p_{i_\ell},
\end{align}
with \(\mathcal{D}(G)\), \(p_{i_\ell}\) and \(\mathcal E_{i_\ell}(t)\) defined in~ Definition \ref{def:diameter}, (\ref{eq:bellman}) and (\ref{eq:mathcal-E}), respectively.
\end{theorem}

\begin{proof}
Let \( \tilde{G} = (V, E) \) be a modified graph that shares the same structure as graph \(G\) but with edge weights defined by \( \tilde{w}_{ij} = (1 - \alpha_1) w_{ij} \). Let \( \tilde{u}_{ij}(t) = \alpha_1 w_{ij} + u_{ij}(t) \), \( \tilde{p}_i = (1 - \alpha_1) p_i \), and \( \tilde{e}_i(t) = x_i(t) - \tilde{p}_i \). Then, we have \( 0 \leq \tilde{u}_{ij}(t) \leq (\alpha_1 + \alpha_2) w_{ij}\). We further define \(\tilde{u}_{ij}^+ = (\alpha_1 + \alpha_2) w_{ij}\).

Since all edge weights are uniformly scaled by the same factor, the shortest paths of node \(i\) and the associated node sequence \( \{i_0, i_1, \dots, i_\ell = i\} \) in \( \tilde{G} \) remain unchanged as in graph \(G\). Consequently, \(\tilde{p}_i\) is the length of the shortest path from \(i\) to its nearest source in graph \(\tilde{G}\), and \(\mathcal{D}(\tilde{G})\), the effective diameter of \(\tilde{G}\), is still \(\mathcal{D}(G)\). Furthermore, since \( \tilde{p}_i \leq p_i \), Assumption~\ref{ass:overestimated} also holds in \( \tilde{G} \). 

Therefore, Lemmas~\ref{lem:upper} and~\ref{lem:lower} can be directly applied to \( \tilde{G} \), i.e., for all \(\ell \in \{1, \cdots, \mathcal{D}(G) -1 \}\),
\begin{equation}
0\leq\tilde e_{i_\ell}(t)\leq \mathcal E_{i_\ell}(t)+\sum_{k=1}^\ell \tilde u_{i_ki_{k-1}}^+.\label{eq:thm1:eq1}
\end{equation}
Substituting 
\(\sum_{k=1}^\ell \tilde{u}_{i_k i_{k-1}}^+ = (\alpha_1 + \alpha_2) p_{i_\ell}\) 
and 
\(\tilde{e}_i(t) = e_i(t) + \alpha_1 p_i\) 
into~\eqref{eq:thm1:eq1} yields (\ref{eq:bound}).
\end{proof}

We proceed to analyze the general case as described in Assumption~\ref{ass:bound}. To this end, we define a special graph \( G^{-} \) as follows.

\begin{definition}
\label{def:G-}
Let \( G^{-} = (V, E) \) denote a modified graph that shares the same topology as \( G \), but with edge weights defined by
\( w_{ij}^- = w_{ij} - u_{ij}^- \), with \( u_{ij}^- \) defined in Assumption~\ref{ass:bound}.  
Since \( G^{-} \) satisfies Assumption~\ref{ass:Graph}, we define \( p_i^- \) as the length of the shortest path from node \( i \) to the source set \( S \) in \( G^{-} \).
\end{definition}

\begin{theorem}
\label{thm:general}
Consider (\ref{eq:pert-DBMC}). For all \( \ell \in \{1, \cdots, \mathcal{D}(G)-1 \} \),
\begin{align}
- (\mathcal{D}({G}^-)-1) u^{-} \leq e_{i_\ell}(t) \leq \ell u^{+} + \mathcal{E}_{i_\ell}(t), \label{eq:thm2:estimate}
\end{align}
where \( u^- , u^+ \) are defined in (\ref{eq:uniform}), \( \mathcal{D}(\cdot) \) in Definition~\ref{def:diameter}, \( \mathcal{E}_{i_\ell}(t) \) in (\ref{eq:mathcal-E}), and \( G^- \) in Definition~\ref{def:G-}.
\end{theorem}
\begin{proof}
The upper bound follows directly from Lemma~1. As for the lower bound, consider the nominal DBMC~\eqref{eq:nominal-DBMC} on graph \(G^-\), and let \(y_i(t)\) denote the state of node \(i\) with the same initial condition as in the perturbed DBMC~\eqref{eq:pert-DBMC}. Since \( {w}_{ij}^- \leq w_{ij}(t) \), the comparison principle~\cite{khalil2002} implies \( y_i(t) \leq x_i(t) \), with \(x_i(t)\) defined in~\eqref{eq:pert-DBMC}. Moreover, since Assumption~\ref{ass:overestimated} also holds in \( G^- \) under~\eqref{eq:nominal-DBMC}, applying Lemma~2 yields 
\begin{align}
x_i(t) \geq y_i(t) \geq p_i^-.\label{eq:thm2:estimateeq1}
\end{align}

Next, consider the shortest path of \( i \) to its nearest source, denoted by \(\{ i_0, i_1, \dots, i_{\ell} = i \}\), with \( i_0 \in S_1 \) and \( \ell \leq \mathcal{D}({G}^-)-1 \). Since \( p_i \) denotes the shortest path length of \(i\) in \( G \) and \(w_{ij}^-\geq w_{ij}-u^-\) by~\eqref{eq:uniform}, we have
\begin{align}
{p}_i^- = \sum_{k=1}^\ell {w}_{i_k i_{k-1}}^- \notag
&\geq \sum_{k=1}^\ell w_{i_k i_{k-1}} - (\mathcal{D}({G}^-)-1) u^- \\
&\geq p_i - (\mathcal{D}({G}^-)-1) u^-. \label{eq:thm2:estimateeq2}
\end{align}

Substituting~\eqref{eq:thm2:estimateeq2} into~\eqref{eq:thm2:estimateeq1} yields \( x_i(t) \geq p_i - (\mathcal{D}(G^-)-1) u^- \), which completes the proof.
\end{proof}

\subsection{Finite-Gain Behavior under Early Termination}
Although \cite{article3} proves that \eqref{eq:nominal-DBMC} achieves predefined finite time stability, ensuring \( x_i(t) = p_i \) for all \( t \geq T_s \), this result hinges on a stringent condition that \( \dot{x}_i(t) \to 0 \) as \( t \to T_s^- \), which generally does not hold in the presence of disturbances.

To address this issue, this subsection introduces a time \( t_s \), slightly earlier than the predefined time instant \( T_s \), such that terminating the disturbed DMBC \eqref{eq:pert-DBMC} at \( t_s \) yields the shortest path for each non-source node. Specifically, each path is constructed according to Bellman’s optimality principle~\eqref{eq:bellman}, by recursively tracing the current parent node of a given node (see the path reconstruction procedure described in Algorithm~1 of~\cite{article1} for further details). Towards this end, we first introduce the following definition to determine whether the perturbed DBMC~\eqref{eq:pert-DBMC} correctly identifies the shortest path at time \( t_s \).

\begin{definition}
We say that the perturbed DBMC (\ref{eq:pert-DBMC}) correctly identifies the shortest path for all \( i \in S_2 \) at time \( t_s \) if every current parent node of \( i \) at time \( t_s \) is also a true parent node of \( i \). That is,
\begin{align}
\mathcal{P}_i(t_s) \subset \mathcal{P}_i, \quad \forall i \in S_2,
\label{eq:path}
\end{align}
with \( \mathcal{P}_i \) and \( \mathcal{P}_i(t_s) \) defined in Definitions~\ref{def:parent} and~\ref{def:currentparent}, respectively.
\end{definition}
Note that when~\eqref{eq:path} is satisfied, each non-source node indeed finds the shortest path towards the nearest source, even though the associated path length estimate may be inaccurate.

To facilitate the subsequent analysis, we further introduce a mild and practical assumption, which is assumed to hold throughout this subsection.
\begin{assumption}
  We assume that the graph \( G = (V, E) \) admits a minimum path length gap \( \zeta > 0 \), defined as the smallest difference between the length of the shortest path and that of the second shortest path, taken over all non-source nodes \( i \in S_2 \).
\label{def:zeta}
\end{assumption}
\begin{remark}
Assumption~\ref{def:zeta} is typically satisfied in practical scenarios or implementations. For instance, in hop-count-based networks~\cite{hop-count1,hop-count2}, we have \( \zeta = 1 \). Similarly, in digital implementations where all edge weights are represented with at most \( d \) decimal digits, \( \zeta \) can be set to \( 10^{-d} \).
\end{remark}

The following lemma provides a sufficient condition for the correct identification of any shortest path.

\begin{lemma}
Consider the disturbance \( u_{ij}(t) \) as specified in Assumption \ref{ass:bound}, suppose \( u^- + u^+ < \zeta \) with \( \zeta \) defined in Assumption~\ref{def:zeta} and \(u^-, u^+\) defined in (\ref{eq:uniform}). If \( |e_i(t_s)| < (\zeta - u^- - u^+)/2 \) for all \( i \in V \), then the perturbed DBMC (\ref{eq:pert-DBMC}) correctly identifies the shortest path for each non-source node at time \( t_s \).
\end{lemma}

\begin{proof}

Consider any \( j \notin \mathcal{P}_i \) and any \( j^\ast \in \mathcal{P}_i \) with \(i \in S_2\). To ensure correct path identification, it suffices to show that \(x_j(t_s) + w_{ij}(t_s) > x_{j^\ast}(t_s) + w_{ij^\ast}(t_s)\), or equivalently,
\begin{align}
p_j + e_j(t_s) + w_{ij} + u_{ij}(t_s) > p_{j^\ast} + e_{j^\ast}(t_s) + w_{ij^\ast} + u_{ij^\ast}(t_s). \label{eq:lem3:eq2}
\end{align}

To ensure~\eqref{eq:lem3:eq2} holds for all admissible disturbances satisfying \( -u^- \leq u_{ij}(t), u_{ij^\ast}(t) \leq u^+ \), a sufficient condition is
\begin{align}
p_j + w_{ij} - p_{j^\ast} - w_{ij^\ast} > e_{j^\ast}(t_s) - e_j(t_s) + u^- + u^+. \label{eq:lem3:eq3}
\end{align}

By the definition of \( \zeta \), we have
\begin{align}
p_j + w_{ij} - p_{j^\ast} - w_{ij^\ast} \geq \zeta.
\end{align}

Therefore, inequality~\eqref{eq:lem3:eq3} is guaranteed if
\begin{align}
\zeta > e_{j^\ast}(t_s) - e_j(t_s) + u^- + u^+. \label{eq:lem3:eq4}
\end{align}

To ensure this for all \( i \in S_2 \), we require \(|e_i(t_s)| < (\zeta - u^- - u^+)/2,\) which completes the proof.
\end{proof}

The following lemma provides a tractable upper bound for \( \mathcal{E}_{i_\ell}(t) \) defined in~\eqref{eq:mathcal-E}, which will be used to calculate \( t_s \).

\begin{lemma}
Consider \(\mathcal E_{i_\ell}(t)\) defined in (\ref{eq:mathcal-E}), there holds
\begin{align}
\mathcal E_{i_\ell}(t) <
\chi_0\dfrac{q^\ell-1}{q-1}\left(\dfrac{T_s - t}{T_s}\right)^{(2h+2)(1-1/q)}, \label{eq:lem4:estimate}
\end{align}
where \( \chi_0 \) denotes the maximum initial state error among all nodes, and \( q > 1 \) is an arbitrary constant.
\end{lemma}

We now present a condition under which the perturbed DBMC~\eqref{eq:pert-DBMC} correctly identifies the shortest path when terminated at an appropriately selected time \( t_s \).

\begin{theorem}
\label{thm:ts}
Consider (\ref{eq:pert-DBMC}), with \(u^-, u^+\) and \(u_{ij}(t) \) defined in (\ref{eq:uniform}). Suppose the following condition holds:
\begin{align}
(u^- + u^+)/2 + u_g < \zeta/2,\label{eq:thm3:condition1}
\end{align}
where \( \zeta \) is defined in Assumption~\ref{def:zeta}, and \( u_g \) is given by
\begin{align}
u_g = \max\left\{ (\mathcal{D}(G^-) - 1) u^-,\, (\mathcal{D}(G) - 1) u^+ \right\},
\end{align}
with \( \mathcal{D}(\cdot) \) and   \( G^- \) defined in Definitions~\ref{def:diameter} and~\ref{def:G-}, respectively.
Then the perturbed DBMC (\ref{eq:pert-DBMC}) correctly identifies the shortest path at time \( t_s \), if
\begin{align}
t_s \geq T_s \left( 1 - \sqrt[(2h + 2)(1 - 1/q)]{\frac{(\zeta - u^- - u^+)/2 - u_g}{\chi_0(q^{\mathcal D(G)-1}-1)/(q-1)}} \right),
\label{eq:thm3:condition}
\end{align}
where \( q > 1 \) is an arbitrary constant.
\end{theorem}

\begin{proof}

Substituting the estimate~\eqref{eq:lem4:estimate} from Lemma~4 into Theorem~2 and noting that \( \ell \leq {\mathcal{D}(G)} - 1 \), we obtain
\begin{align}
|e_{i_{\ell}}(t_s)| 
<u_g
+ \chi_0\dfrac{q^{\mathcal D(G)-1}-1}{q-1} \left( \frac{T_s - t_s}{T_s} \right)^{(2h+2)(1 - 1/q)}.
\end{align}
Since Lemma~3 requires \( |e_i(t_s)| < (\zeta-u^--u^+)/2 \) for the perturbed DBMC~\eqref{eq:pert-DBMC} to correctly determine the shortest path, it suffices to ensure the following condition:
\begin{align}
u_g
+ \chi_0\dfrac{q^{\mathcal D(G)-1}-1}{q-1} \left( \frac{T_s - t_s}{T_s} \right)^{(2h+2)(1 - 1/q)}\leq \dfrac{\zeta-u^--u^+}2.\label{eq:thm3:eq1}
\end{align}
Solving~\eqref{eq:thm3:eq1} yields~\eqref{eq:thm3:condition}, thereby completing the proof.
\end{proof}

\section{Simulation}

This section presents numerical simulations to validate the theoretical results established earlier.
Consider a directed hop-count network consisting of 13 nodes, with node~1 being the sole source node, i.e., \(V = \{1, \cdots, 13\}\) and \(S_1 = \{1 \}\). The parameters in DBMC~\eqref{eq:pert-DBMC} are set as \( h = 12 \), \( T_s = 5 \), and \( \gamma = 2 \). To satisfy Assumption~\ref{ass:overestimated}, the initial states are set as \( x_i(0) = 12 \) for all \(i \in S_2\).

Figure~\ref{fig:state-error} illustrates the evolution of state errors of~\eqref{eq:pert-DBMC} with the edge weight \(w_{ij}(t)\) obeying \(w_{ij}(t) \in [0.6w_{ij}, 1.4w_{ij}]\). It can be seen from Figure~\ref{fig:state-error} that the error remains bounded over \( t \in [0,T_s) \). In particular, Figure~\ref{fig:node-8-error} shows that the state error of node~8 is well contained within the theoretical bounds as established by Theorem~\ref{thm:proportional}.

To verify Theorem~\ref{thm:ts}, we consider a smaller disturbance such that \(w_{ij}(t) \in [0.97w_{ij}, 1.03w_{ij}] \). By setting \(q = 3\) in~\eqref{eq:lem4:estimate}, the early termination time defined in Theorem~\ref{thm:ts} is given by \(t_s = 3.1445\). Figure~\ref{fig:node-8-path} demonstrates the resulting path (highlighted in red) of node~\(8\) by terminating (\ref{eq:pert-DBMC}) at \(t_s\). It can be readily verified that the DBMC algorithm has successfully identified a shortest path for node~\(8\), though the estimated length slightly deviates from the stationary value. 

\begin{figure}[htbp]
    \centering
    \includegraphics[width=\linewidth]{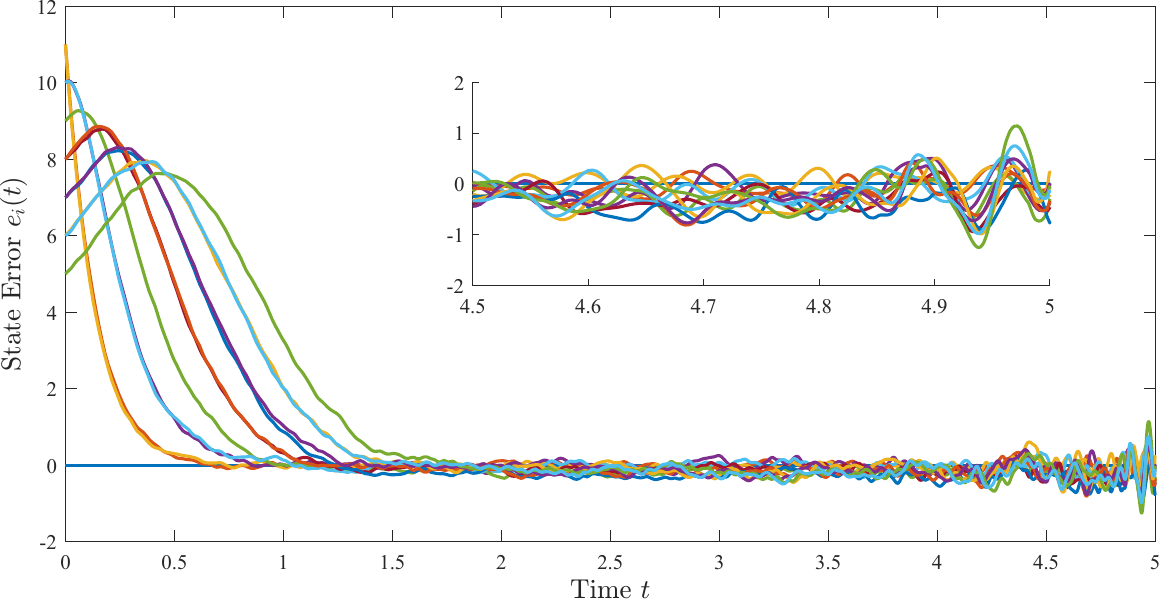}
    \caption{State errors of (\ref{eq:pert-DBMC}) in a hop-count network subjected to \(40\%\) variation in edge weights.}
    \label{fig:state-error}
\end{figure}

\begin{figure}[htbp]
    \centering
    \includegraphics[width=\linewidth]{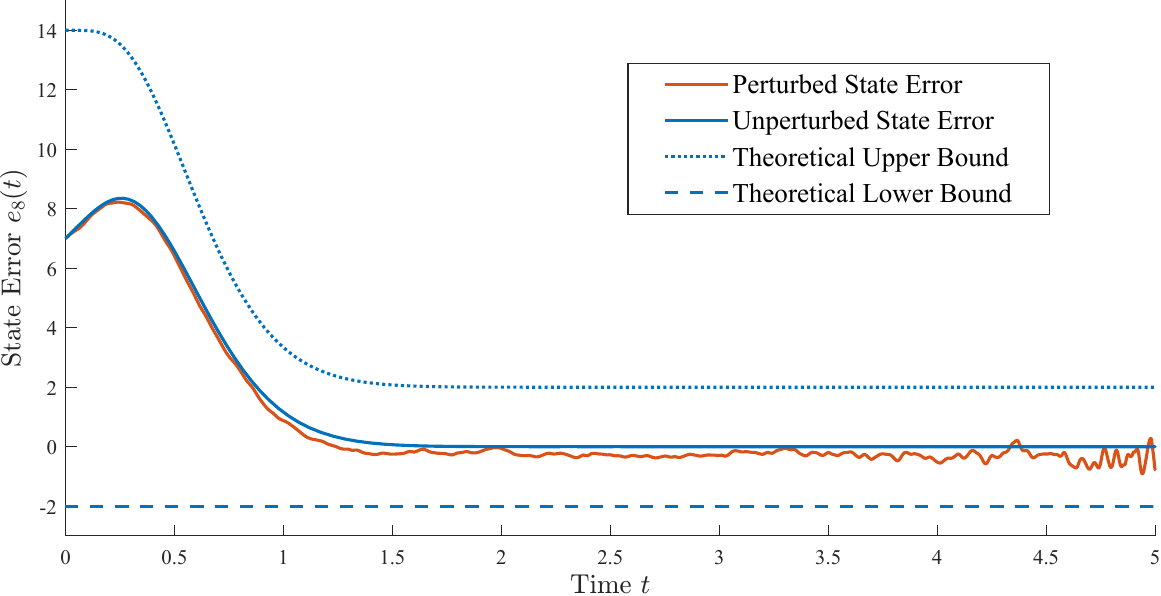}
    \caption{Comparison between the state error of node~\(8\) and the corresponding theoretical bounds.}
    \label{fig:node-8-error}
\end{figure}

\begin{figure}[htbp]
    \centering
    \includegraphics[width=\linewidth]{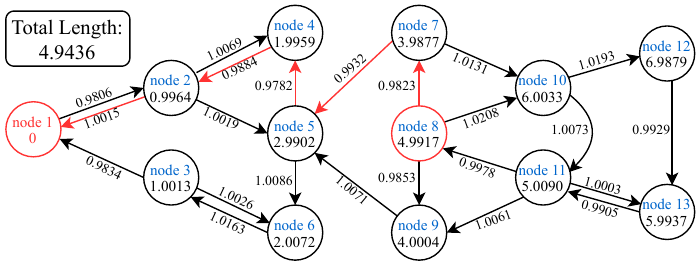}
    \caption{The resulting path of node~\(8\) by terminating~\eqref{eq:pert-DBMC} at \(t_s = 3.1445\) subjected to \(3\%\) variation in edge weights.}
    \label{fig:node-8-path}
\end{figure}

\section{Conclusion}
This paper focuses on DBMC under the PT control strategy proposed in~\cite{article3}, and aims to address the erratic behavior of the DBMC protocol~\eqref{eq:pert-DBMC} near the prescribed time instant, in the presence of disturbances on the edge weights. We first established refined bounds on the state error. Then, by formalizing the notion of correct shortest path identification, we derived a sufficient condition for the early termination time of the perturbed DBMC~\eqref{eq:pert-DBMC} that guarantees shortest path correctness for each node. While the proposed condition is broadly applicable, it may be conservative. Enhancing its tightness, especially under structured graph topologies, remains an open avenue for future research.

\bibliographystyle{IEEEtran}
\bibliography{references}

\appendix

\textbf{Proof of Lemma 1:} We prove (\ref{eq:lem1:eq1}) by mathematical induction. When \( \ell = 0 \), since \( e_{i_0}(t) \equiv 0 \) by~\eqref{eq:pert-DBMC}, both sides of~\eqref{eq:lem1:eq1} equal zero. Hence, the inequality holds trivially. Now assume that~\eqref{eq:lem1:eq1} holds for some \( \ell \in \{0, 1, \cdots, \mathcal{D}(G) - 2 \}  \), we proceed to show \eqref{eq:lem1:eq1} also holds for \( \ell + 1 \).

It follows from ~\eqref{eq:pert-DBMC} and \(i_{\ell} \in \mathcal{N}_{i_{\ell+1}}\) that
\begin{align}
&\dot{e}_{i_{\ell+1}}(t) =\dot{x}_{i_{\ell+1}}(t)\notag
\\=&-\eta(t) ( x_{i_{\ell+1}}(t) - \min_{j\in \mathcal N_{i_{\ell+1}}} \{x_j(t) + w_{i_{\ell+1}j} + u_{i_{\ell+1}j}(t)\} ) \notag\\
\leq&-\eta(t) \left( x_{i_{\ell+1}}(t) - x_{i_{\ell}}(t) - w_{i_{\ell+1} i_{\ell}} - u_{i_{\ell+1} i_{\ell}}(t) \right) \notag \\
=& -\eta(t) \left( (x_{i_{\ell+1}}(t)-p_{i_{\ell+1}}) - (x_{i_{\ell}}(t) - p_{i_{\ell}}) - u_{i_{\ell+1} i_{\ell}}(t) \right) \label{eq:lem1:eq2}\\
\leq& -\eta(t) ( e_{i_{\ell+1}}(t) - e_{i_{\ell}}(t) - u_{i_{\ell+1} i_{\ell}}^+).\label{eq:lem1:eq3}
\end{align}
where \eqref{eq:lem1:eq2} uses \(  p_{i_{\ell+1}} = w_{i_{\ell+1}i_{\ell}} + p_{i_{\ell}} \) since \( i_{\ell+1} \in \mathcal{P}_{i_{\ell}} \).

Substituting the induction hypothesis~\eqref{eq:lem1:eq1} into~\eqref{eq:lem1:eq3}, we obtain
\begin{align}
&\dot e_{i_{\ell+1}}(t)+\eta(t)e_{i_{\ell+1}}(t)\notag\\
\leq& \eta(t)\sum_{k=0}^{\ell}  \frac{e_{i_k}(0)\, (\ln \phi(t))^{\ell - k}}{\phi(t)\,(\ell - k)!} + 
\eta(t) \sum_{k=0}^{\ell+1} u_{i_k i_{k-1}}^{+}.\label{eq:lem1:eq4}
\end{align}

Inequality~\eqref{eq:lem1:eq4} is essentially a first-order differential inequality that can be solved analytically. By applying the comparison principle~\cite{khalil2002} and recalling the definition of \(\phi(t)\) in~\eqref{eq:phi}, we have
\begin{align}
e_{i_{\ell+1}}(t) \leq & \frac{1}{\phi(t)} \Bigg( \int_0^t \Big(
 \eta(s) \sum_{k=0}^{\ell} \frac{e_{i_k}(0)\, (\ln \phi(s))^{\ell -k}}{\phi(s)\,(\ell -k)!} + \notag\\
& \eta(s) \sum_{k=0}^{\ell+1} u_{i_k i_{k-1}}^{+}
\Big) \phi(s) \,\mathrm{d}s + e_{i_{\ell+1}}(0) \Bigg)\notag\\
=&\underbrace{\frac{1}{\phi(t)} \sum_{k=0}^{\ell} \int_0^t   \frac{e_{i_k}(0)\, (\ln \phi(s))^{\ell -k}}{\phi(s)\,(\ell -k)!} \mathrm d \phi(s)}_{\mathcal I(t)} + \notag\\
& \dfrac{1}{\phi(t)}\sum_{k=0}^{\ell+1}u_{i_k i_{k-1}}^+ \int_0^t\,\mathrm{d}\phi(s) +\dfrac{e_{i_{\ell+1}}(0)}{\phi(t)},\label{eq:lem1:eq6}
\end{align}
where \eqref{eq:lem1:eq6} uses \( \mathrm{d} \phi(s) = \eta(s)\phi(s)\,\mathrm{d}s\).

Since \(\phi(0)=1\), we have
\begin{align}
\dfrac{1}{\phi(t)}\int_0^t\,\mathrm{d}\phi(s)=\dfrac{\phi(t)-\phi(0)}{\phi(t)}\leq 1,\label{eq:lem1:eq7}
\end{align}

and 
\begin{align}
\mathcal I(t)=&\dfrac{1}{\phi(t)}\sum_{k=0}^{\ell}\int_0^t   \frac{e_{i_k}(0)\, (\ln \phi(s))^{\ell-k}}{(\ell -k)!} \mathrm d \ln\phi(s)\notag\\
=&\sum_{k=0}^{\ell}\dfrac{e_{i_k}(0)}{\phi(t)}\dfrac{(\ln\phi(t))^{\ell+1-k}}{(\ell+1-k)!}.\label{eq:lem1:eq8}
\end{align}
Substituting~\eqref{eq:lem1:eq7} and~\eqref{eq:lem1:eq8} into~\eqref{eq:lem1:eq6} produces the bound for \( e_{i_{\ell+1}}(t) \) in the same form as~\eqref{eq:lem1:eq1}, completing the proof.

\textbf{Proof of Lemma 4:}  
Using the inequality \( \ln x \leq x/e \) and the Stirling-type lower bound~\cite{robbins1955stirling}, i.e., \( k! \geq \sqrt{2\pi k}(k/e)^k>(k/e)^k \), we obtain that for all $k \geq 1$,
\begin{align}
\dfrac{(\ln \phi(t))^k}{\phi(t) k!}
= \dfrac{\left(qk \ln(\phi(t)^{1/(qk)})\right)^k}{\phi(t) k!}\\
< \dfrac{\left(qk \phi(t)^{1/(qk)}/e \right)^k}{\phi(t)(k/e)^k}
= q^k\, \phi(t)^{1/q - 1}, \label{eq:stirling-estimate}
\end{align}
whenever \(q> 1\). When \(k=0\), \(1/\phi(t)< q^0 \phi(t)^{1/q-1}\) holds trivially. Therefore, for all $k \geq 0$, substituting the expression of $\phi(t)$ from~\eqref{eq:phi} yields
\begin{align}
\dfrac{(\ln \phi(t))^k}{\phi(t) k!}
&< q^k\, e^{\gamma t(1/q - 1)} \left( \dfrac{T_s - t}{T_s} \right)^{(2h + 2)(1 - 1/q)} \notag\\
&\leq q^k \left( \dfrac{T_s - t}{T_s} \right)^{(2h + 2)(1 - 1/q)}.\label{eq:lemma4:eq2}
\end{align}

Now, applying~\eqref{eq:lemma4:eq2} to~\eqref{eq:mathcal-E} and noting that \( e_{i_k}(0) \leq \chi_0 \), we obtain
\begin{align}
\mathcal E_{i_\ell}(t) <&{\chi_0}\left(\dfrac{T_s-t}{T_s}\right)^{(2h+2)(1-1/q)}\sum_{k=1}^{\ell} q^{\ell-k}\notag\\
=&{\chi_0}\dfrac{q^{\ell}-1}{q-1}\left(\dfrac{T_s-t}{T_s}\right)^{(2h+2)(1-1/q)},
\end{align}
\noindent
which completes the proof.
\end{document}